\documentclass[11pt,a4paper]{article}

\usepackage{fullpage}
\usepackage{amsmath,latexsym,amssymb} 
\usepackage{amsthm}
\usepackage{algorithm,algorithmic}

\usepackage{color,soul}
\usepackage{mdwlist,paralist}

\def\cG{{\cal G}}
\def\DEF{\stackrel{\rm def}{=}}

\newtheorem{theorem}{Theorem}
\newtheorem{lemma}{Lemma}

\newtheorem{definition}{Definition}
\newtheorem{corollary}[theorem]{Corollary}
\newtheorem{observation}[lemma]{Observation}

\newcommand{\set}[1]{\left\{ #1 \right\}}
\newcommand{\paren}[1]{\left( #1 \right)}
\newcommand{\inv}[1]{\frac{1}{#1}}

\newcommand{\ceil}[1]{\left\lceil {#1} \right\rceil}
\newcommand{\floor}[1]{\left\lfloor {#1} \right\rfloor}

\newcommand{\reals}{\mathbb{R}}

\newcommand{\argmin}{\operatornamewithlimits{argmin}}
\newcommand{\argmax}{\operatornamewithlimits{argmax}}
\newcommand{\eps}{\varepsilon}

\newcommand{\opt}{\textsc{opt}}

\newcommand{\appr}{\textsc{appr}}

\def\mvbp{\textsc{mvbp}}
\def\mmk{\textsc{mmk}}

\newcommand{\knapsack}{\textsc{knapsack}}
\newcommand{\bp}{\textsc{bp}}
\newcommand{\vbp}{\textsc{vbp}}
\newcommand{\setcover}{\textsc{sc}}

\newcommand{\ff}{\textsc{ff}}

\newcommand{\cA}{\mathcal{A}}
\newcommand{\cC}{\mathcal{C}}

\newcommand{\bl}{\bar{\ell}}

\renewcommand{\paragraph}[1]{\medskip\par\noindent\textbf{#1} }

\begin{document}

\title{\textbf{Vector Bin Packing with Multiple-Choice}%
\thanks{Research supported in part by the Next Generation Video 
(NeGeV) Consortium, Israel.} \\
\textsc{\Large Extended Abstract}
}

\author{
\begin{tabular}{c@{\extracolsep{15pt}}c}
Boaz Patt-Shamir%
\thanks{Supported in part by the Israel Science Foundation
 (grant 664/05).}
 & Dror Rawitz\\
{\small\tt boaz@eng.tau.ac.il}& {\small\tt rawitz@eng.tau.ac.il}\\
\multicolumn{2}{c}{School of Electrical Engineering}\\
\multicolumn{2}{c}{Tel Aviv University}\\
\multicolumn{2}{c}{Tel Aviv~~69978}\\
\multicolumn{2}{c}{Israel}\\
\end{tabular}
}

\date{}

\begin{titlepage}
\maketitle

\begin{abstract}
We consider a variant of \emph{bin packing} called
\emph{multiple-choice vector bin packing}.  In this problem we are
given a set of items, where each item can be selected in one of
several $D$-dimensional \emph{incarnations}.  We are also given $T$
bin types, each with its own cost and $D$-dimensional size.  Our goal
is to pack the items in a set of bins of minimum overall cost.
The problem is motivated by scheduling in networks with guaranteed
quality of service (QoS), but due to its general formulation it has
many other applications as well.
We present an approximation algorithm that is guaranteed to produce a
solution whose cost is about $\ln D$ times the optimum.  For the
running time to be polynomial we require $D=O(1)$ and $T=O(\log n)$.
This extends previous results for \emph{vector bin packing}, in which
each item has a single incarnation and there is only one bin type.
To obtain our result we also present a PTAS for the multiple-choice
version of \emph{multidimensional knapsack}, where we are given only
one bin and the goal is to pack a maximum weight set of (incarnations
of) items in that bin.
\end{abstract}


\renewcommand{\thepage}{}
\end{titlepage}
\pagenumbering{arabic}

\section{Introduction}

Bin packing, where one needs to pack a given set of items using the
least number of limited-space containers (called bins), is one of the
fundamental problems of combinatorial optimization (see,
e.g.,~\cite{PapadimitriouS}).  In the \emph{multidimensional} flavor
of bin packing, each item has sizes in several dimensions, and the
bins have limited size in each dimension~\cite{knapsack-book}.  In
this paper we consider a natural generalization of multidimensional
bin packing that occurs frequently in practice, namely
\emph{multiple-choice} multidimensional bin packing.  In this variant,
items and space are multidimensional, and in addition, each item may
be selected in one of a few \emph{incarnations}, each with possibly
different sizes in the different dimensions. Similarly, bins can be
selected from a set of types, each bin type with its own size cap in
each dimension, and possibly different cost.  The problem is to select
incarnations of the items and to assign them to bins so that the
overall cost of bins is minimized.

Multidimensionality models the case where the objects to pack have
costs in several incomparable budgets.  For example, consider a
distribution network (e.g., a cable-TV operator), which needs to
decide which data streams to provide.  Streams typically have
prescribed bandwidth requirements, monetary costs, processing
requirements etc., while the system typically has limited available
bandwidth, a bound on the amount of funds dedicated to buying content,
bounded processing power etc.
The multiple-choice variant models, for example, the case where
digital objects (such as video streams) may be taken in one of a
variety of formats with different characteristics (e.g., bandwidth and
processing requirements), and similarly, digital bins (e.g., server
racks) may be configured in more than one way.
The multiple-choice multidimensional variant is useful in many
scheduling applications such as communication under Quality of Service
(QoS) constraints, and including workplans for nursing personnel in
hospitals~\cite{WP-72}.

Specifically, in this paper we consider the problem of
\emph{multiple-choice vector bin packing} (abbreviated \mvbp, see
Section~\ref{Sec:definitions} for a formal definition).  The input to
the problem is a set of $n$ \emph{items} and a set of $T$ \emph{bin
types}.  Each item is represented by at most $m$ \emph{incarnations},
where each incarnation is characterized by a $D$-dimensional vector
representing the size of the incarnation in each dimension.  Each bin
type is also characterized by a $D$-dimensional vector representing
the capacity of that bin type in each dimension.  We are required to
pack all items in the minimal possible number of bins, i.e., we need
to select an incarnation for each item, select a number of required
bins from each type, and give an assignment of item incarnations to
bins so that no bin exceeds its capacity in any dimension.
In the weighted version of this problem each bin type has an associated
cost, and the goal is to pack item incarnations into a set of bins of
minimum cost.

Before stating our results, we note that na\"{\i}ve reductions to the
single-choice model do not work. For example, consider the case where
$n/2$ items can be packed together in a single type-1 bin but require
$n/2$ type-2 bins, while the other $n/2$ items fit together in a
single type-2 bin but require $n/2$ type-1 bins. If one uses only one
bin type, the cost is dramatically larger than the optimum---even with
one incarnation per item.
Regarding the choice of item incarnation, one may try to use only a
cost-effective incarnation for each item (using some natural
definition).  However, it is not difficult to see that this approach
results in approximation ratio $\Omega(D)$ even when there is only one
bin type.

\subsection{Our Results}
In this paper we give a polynomial-time approximation algorithm for
the multiple-choice vector bin packing problem, in the case where $D$
(the number of dimensions) is constant.  The approximation ratio for
the general weighted version is $\ln 2D+3$, assuming that $T$ (the
number of bin types) satisfies $T=O(\log n)$.  For the unweighted
case, the approximation ratio can be improved to $\ln2D+1+\eps$, for
any constant $\eps>0$, if $T=O(1)$ as well.  Without any assumption on
$T$, we can guarantee, in the unweighted case, cost of
$(\ln(2D)+1)\opt+T+1$, where $\opt$ denotes the optimal cost.
To the best of our knowledge, this is the first 
approximation algorithm for the problem with multiple choice, and it
is as good as the best solution for single-choice vector bin packing
(see below).%

As an aside, to facilitate our algorithm we also improve on the best
results for multiple-choice multidimensional \emph{knapsack} problem
(abbreviated \mmk), where we are given a single bin and the goal is to
load it with the maximum weight set of (incarnations of) items.
Specifically, we present a polynomial-time approximation scheme (PTAS)
for \mmk\ for the case where the dimension $D$ is constant.  The
PTAS
for \mmk\ is used as a subroutine in our algorithm for \mvbp.

\subsection{Related Work}
Classical bin packing (\bp) (single dimension, single choice) admits
an asymptotic PTAS~\cite{FerLue81} and an asymptotic fully
polynomial-time approximation scheme (asymptotic
FPTAS)~\cite{KarKar82}.
Friesen and Langston~\cite{FriLan86} presented constant factor
approximation algorithms for a more general version of \bp\ in which a
fixed collection of bin sizes is allowed, and the cost of a solution
is the sum of sizes of used bins.  For more details about this version
of \bp\ see~\cite{SvSE03} and references therein.
Correa and Epstein~\cite{CorEps08} considered \bp\ with controllable
item sizes.  In this version of \bp\ each item has a list of pairs
associated with it.  Each pair consists of an allowed size for this
item, and a nonnegative penalty.  The goal is to select a pair for
each item so that the number of bins needed to pack the sizes plus the
sum of penalties is minimized.  Correa and Epstein~\cite{CorEps08}
presented an asymptotic PTAS that uses bins of size slightly larger
than $1$.

Regarding multidimensionality, it has been long known that vector bin
packing (\vbp, for short) can be approximated to within a factor of
$O(D)$~\cite{GGJY76,FerLue81}.  More recently, Chekuri and
Khanna~\cite{CheKha04} presented an $O(\log D)$-approximation
algorithm for \vbp, for the case where $D$ is constant.  They also
showed that approximating \vbp\ for arbitrary dimension is as hard as
graph coloring, implying that it is unlikely that \vbp\ admits
approximation factor smaller than $\sqrt{D}$.  The best known
approximation ratio for \vbp\ is due to Bansal, Caprara and
Sviridenko~\cite{BCS06}, who gave a polynomial-time approximation
algorithm for constant dimension $D$ with approximation ratio
arbitrarily close to $\ln D + 1$.
Our algorithm for \mvbp\ is based on their ideas.

For the knapsack problem, Frieze and Clarke~\cite{FriCla84} presented
a PTAS for the (single-choice) multidimensional variant, but obtaining
an FPTAS for multidimensional knapsack is NP-hard~\cite{MC-84}.
Shachnai and Tamir~\cite{ShaTam03} use the approach of~\cite{FriCla84}
to obtain a PTAS for a special case of 2-dimensional multiple-choice
knapsack.  Our algorithm for \mmk\ extends their technique to the
general case.
\mmk\ was studied extensively by practitioners.  Heuristics for \mmk\ 
abound, see, e.g.,~\cite{Khan,HMS-04,PD-05,ARKMS06,S-07}.  From the
algorithmic viewpoint, the first relevant result is by Chandra et
al.~\cite{CHW76}, who present a PTAS for single-dimension,
multiple-choice knapsack.

\subsection{Paper Organization}
The remainder of this paper is organized as follows. 
In Section~\ref{Sec:definitions} we formalize the problems. In
Section~\ref{Sec:knapsack} we present our solution to the \mmk\
problem, which is used in our solution to the \mvbp\ problem that is
presented in Section~\ref{Sec:mvbp}.

\section{Problem Statements}
\label{Sec:definitions}

We now formalize the optimization problems we deal with.  For a natural
number $n$, let $[n]\DEF\set{1,2,\ldots,n}$ (we use this notation
throughout the paper).

\paragraph{Multiple-Choice Multidimensional Knapsack problem} (\mmk).
\begin{description*}
\item[\textit{Instance:}] A set of $n$ \emph{items}, where each item
      is a set of $m$ or fewer $D$-dimensional \emph{incarnations}.
      Incarnation $j$ of item $i$ has size $a_{ij}\in(\reals^+)^D$, in
      which the $d$th dimension is a real number $a_{ijd}\ge 0$. \\
      In the \emph{weighted} version, each incarnation $j$ of item $i$
      has weight $w_{ij} \geq 0$.
\item[\textit{Solution:}] 
      A set of incarnations, at most one of each item, such that the
      total size of the incarnations in each dimension $d$ is at most
      $1$.
\item[\textit{Goal:}] Maximize the number (\emph{weighted version:}
      total weight) of incarnations in a solution.
\end{description*}
When $D=m=1$, this is the classical Knapsack problem (\knapsack).

\bigskip

\paragraph{Multiple-Choice Vector Bin Packing} (\mvbp).
\begin{description*}
\item[\textit{Instance:}] Same as for unweighted \mmk, with the addition 
      of $T$ \emph{bin types}, where each bin type $t$ is
      characterized by a vector $b_t \in (\reals^+)^D$. The $d$th
      coordinate of $b_t$ is called the \emph{capacity} of type $t$ in
      dimension $d$, and denoted by $b_{td}$.  \\
In the \emph{weighted}
      version, each bin type $t$ has a weight $w_t\ge0$.
\item[\textit{Solution:}] A set of bins, each assigned a bin type and
  a set of item incarnations, such that exactly one
      incarnation of each item is assigned to any bin, and such that the total
      size of incarnations assigned to a bin does not exceed its
      capacity in any dimension
\item[\textit{Goal:}] Minimize number of (\emph{weighted version:}
      total weight of) assigned bins. 
\end{description*}
When $m=1$ we get \vbp, and the special case where $D=m=1$ is the
classical bin packing problem (\bp).

\section{Multiple-Choice Multidimensional Knapsack}
\label{Sec:knapsack}

In this section we present a PTAS for weighted \mmk\ for the case
where $D$ is a constant.  Our construction extends the algorithms of
Frieze and Clarke~\cite{FriCla84} and of Shachnai and
Tamir~\cite{ShaTam03}.

We first present a linear program of \mmk, where the variables
$x_{ij}$ indicate whether the $j$th incarnation of the $i$th item is
selected.
\begin{align}
\begin{array}{ll@{\hspace{10pt}}l}
\max        
& \displaystyle 
  \sum_{i=1}^n \sum_{j=1}^m w_{ij} x_{ij} \\[12pt]
\text{s.t.} 
& \displaystyle 
  \sum_{i=1}^n \sum_{j=1}^m a_{ijd} x_{ij} \leq 1 &\forall d \in [D] \\[12pt]
& \displaystyle
  \sum_{j=1}^m x_{ij} \leq 1                  & \forall i \in [n] \\[12pt]
& x_{ij} \geq 0                               & \forall i \in [n], j \in [m]
\end{array}
\tag{MMK}
\label{LP:MKA}
\end{align}
In the program, the first type of constraints make sure that the load on
the knapsack in each dimension is bounded by $1$; the second type of
constraints ensures that at most one copy of each element is taken
into the solution. Constraints of the third type indicate the
relaxation: the integer program for \mmk\ requires that $x_{ij} \in
\set{0,1}$.

Our PTAS for \mmk\ is based on the linear program~\eqref{LP:MKA}.  Let
$\eps>0$. Suppose we somehow guess the heaviest $q$ incarnations that
are packed in the knapsack by some optimal solution, for $q = \min
\set{n,\ceil{D/\eps}}$.  Formally, assume we are given a set $G
\subseteq [n]$ of at most $q$ items and a function $g: G \to [m]$
that selects incarnations of items in $G$.  In this case we can assign
values to some variables of \eqref{LP:MKA} as follows:
\[
x_{ij} =
\begin{cases}
1~, & \mbox{if } i \in G \mbox{ and } j = g(i) \\
0~, & \mbox{if } i \in G\mbox{ and } j \neq g(i) \\
0~, & \mbox{if } i \not\in G\mbox{ and } w_{ij} > \min\{w_{\ell
  g(\ell)}\mid {\ell \in G} \} \\
\end{cases}
\]
That is, if we guess that incarnation $j$ of item $i$ is in the
optimal solution, then $x_{ij} = 1$ and $x_{ij'} = 0$ for $j'\ne j$;
also, if the $j$th incarnation of item $i$ weighs more than some
incarnation in our guess, then $x_{ij} = 0$ .  Denote the resulting
linear program \ref{LP:MKA}$(G,g)$.

Let $x^*(G,g)$ be an optimal (fractional) solution
of~\ref{LP:MKA}$(G,g)$. The idea of Algorithm~\ref{Alg:MKA} below is
to simply round down the values of $x^*(G,g)$. We show that if $G$ and
$g$ are indeed the heaviest incarnations in the knapsack, then the
rounded-down solution is very close to the optimum.  Therefore, in the
algorithm we loop over all possible assignments of $G$ and $g$ and
output the best solution.

\begin{algorithm}[h]
\begin{small}
\caption{}
\label{Alg:MKA}
\begin{algorithmic}[1]
\FORALL{$G \subseteq [n]$ such that $|G| \leq q$
        and $g: G \to [m]$}
  \STATE $b_d(G,g) \leftarrow 1 - \sum_{i \in G} a_{i g(i) d}$ 
         for every $d \in [D]$
  \IF{$b_d(G,g) \geq 0$ for every $d$}
    \STATE Compute an optimal basic solution $x^*(G,g)$ of~\ref{LP:MKA}$(G,g)$
    \STATE $x_{ij}(G,g) \leftarrow \floor{x^*_{ij}(G,g)}$ 
           for every $i$ and $j$
  \ENDIF
  \STATE $x \leftarrow \argmax_{x(G,g)} w \cdot x(G,g)$
\ENDFOR
\RETURN $x$
\end{algorithmic}
\end{small}
\end{algorithm}

\begin{theorem}
\label{The:MKA}
If $D = O(1)$, then Algorithm~\ref{Alg:MKA} is a PTAS for \mmk.
\end{theorem}
\begin{proof}
Regarding running time, note that there are $O(n^q)$ choices of $G$,
and  $O(m^q)$ choices of $g$ for each choice of $G$, and  hence, the
algorithm runs for $O((nm)^q) = O((nm)^{\ceil{D/\eps}})$ iterations,
i.e., time polynomial in the input length, for constant $D$ and $\eps$.

Regarding approximation, fix be an optimal integral solution $x^I$
to~\eqref{LP:MKA}.  If $x^I$ assigns at most $q$ incarnations of items
to the knapsack, then we are 
done.
Otherwise, let $G^I$ be the set of items that correspond to the $q$
heaviest incarnations selected to the knapsack by $x^I$.  For $i \in
G^I$, let $g^I(i)$ denote the incarnation of $i$ that was put in the
knapsack by $x^I$.  Consider the iteration of Algorithm~\ref{Alg:MKA}
in which $G=G^I$ and $g=g^I$.  Clearly, $w \cdot x^I \leq w \cdot
x^*(G^I,g^I)$.
Let $n'$ denote the number of items that were not chosen by
$(G^I,g^I)$ or were eliminated because their incarnations weigh too
much.  Let $k$ be the number of variables of the form $x_{ij}$ in
\ref{LP:MKA}$(G^I,g^I)$.  When using slack form, we add $D + n'$ slack
variables: each constraint of the first type is written as
$\sum_{i=1}^n \sum_{j=1}^m a_{ijd} x_{ij} +s_d= 1$, for $1\le d\le D$,
and each constraint of the second type is written as $\sum_{j=1}^m
x_{ij}+s'_{i}= 1$ for $1\le i\le n'$. Thus, the total number of
variables in the program to $k+ D + n'$,and since
\ref{LP:MKA}$(G^I,g^I)$ has $D+n'$ constraints (excluding positivity
constraints $x_{ij} \geq 0$), it follows any basic solution
of~\ref{LP:MKA}$(G^I,g^I)$ has at most $D+n'$ positive variables.
Since by constraints of the second type in~\ref{LP:MKA}$(G^I,g^I)$
there is at least one positive variable for each item, it follows that
$x^*(G^I,g^I)$ has at most $D$ non-integral entries, and therefore, by
rounding down $x^*(G^I,g^I)$ we lose at most $D$ incarnations of
items. Let $W^I = \sum_{i \in G^I} w_{i g^I(i)}$.  Then each
incarnation lost due to rounding weighs at most $W^I/q$ (because it is
not one of the $q$ heaviest).  We conclude that
\[
\textstyle
w \cdot x(G^I,g^I) 
~\geq~ w \cdot x^*(G^I,g^I) - D \cdot \frac{W^I}{q} %
~\geq~ w \cdot x^*(G^I,g^I) (1 - \frac{D}{q}) %
~\geq~ w \cdot x^I (1 - \frac{D}{q}) %
~=~    \frac{\opt}{1 + \eps}
~,
\]
and we are done.
\end{proof}

\section{Multiple-Choice Vector Bin Packing}
\label{Sec:mvbp}

In this section we present our main result, namely, an $O(\log
D)$-approximation algorithm for \mvbp, assuming that $D$ and $T$
(number of dimensions and bin types, respectively) are constants.  Our
algorithm is based on and extends the work of~\cite{BCS06}.

The general idea is as follows.  We first encode \mvbp\ using a
covering linear programming formulation with exponentially many
variables, but polynomially many constraints.  We find a near optimal
fractional solution to this (implicit) program using a separation
oracle of the dual program.  (The oracle is implemented by the \mmk\
algorithm from Section~\ref{Sec:knapsack}.)
We assign some incarnations to bins using a greedy rule based on some
``well behaved'' dual solution (the number of greedy assignments
depends on the value of the solution to the primal program).  Then we
are left with a set of unassigned items, but due to our greedy rule we
can assign these remaining items to a relatively small number of bins.

\subsection{Covering Formulation}

We start with the transformation of \mvbp\ to weighted Set Cover
(\setcover).  An instance of \setcover\ is a family of sets $\cC =
\set{C_1,C_2,\ldots}$ and a cost $w_C\ge0$ for each $C\in\cC$. We
call $\bigcup_{C \in \cC} C$ the \emph{ground set} of the instance,
and usually denote it by $I$.  The goal in \setcover\ is to choose
sets from $\cC$ whose union is $I$ and whose overall cost is minimal.
Clearly, \setcover\ is equivalent to the following integer program:
\begin{align}
\begin{array}{ll@{\hspace{10pt}}l}
\min 
& \displaystyle 
  \sum_{C \in \cC} w_C \cdot x_C          \\[12pt]
\text{s.t.}     
& \displaystyle
  \sum_{C \ni i} x_C \geq 1      & \forall i \in I \\[12pt]
& x_C \in \set{0,1}              & \forall C \in \cC
\end{array}
\tag{P}
\label{IP:cover}
\end{align}
where $x_C$ indicates whether the set $C$ is in the cover.  A linear
program relaxation is obtained by replacing the integrality
constraints of \eqref{IP:cover} by positivity constraints $x_C \geq 0$
for every $C \in \cC$.  The above formulation is very general. We
shall henceforth call problems whose instances can be formulated as in
\eqref{IP:cover} for some $\cC$ and $w_C$ values, \emph{(P)-problems}.

In particular, \mvbp\ is  a \eqref{IP:cover}-problem, as the
following reduction shows.  Let $\cal I$ be an instance of \mvbp.
Construct an instance $\cC$ of \setcover\ as follows.  The ground set
of $\cC$ is the set of items in $\cal I$, and sets in $\cC$ are the
subsets of items that can be assigned to some bin.
Formally, a set $C$ of items is called \emph{compatible} if and
only if there exists a bin type $t$ and an incarnation mapping $f:C
\to [m]$ such that $\sum_{i \in C} a_{if(i)d} \leq b_{td}$ for every
dimension $d$, i.e., if there is a way to accommodate all members of
$C$ is the same bin.  In the instance of \setcover, we let $\cC$ be
the collection of all compatible item sets.  Note that a solution to
set cover does not immediately solve \mvbp, because selecting
incarnations and bin-types is an NP-hard problem in its own right.  To
deal with this issue we have one variable for each possible
\emph{assignment} of incarnations and bin type.  Namely, we may have
more than one variable for a compatible item subset.
\subsection{Dual Oblivious Algorithms}
\label{Sub:oblivious}

We shall be concerned with approximation algorithms for
\eqref{IP:cover}-problems which have a special property with respect
to the dual program. First, we define the dual to the LP-relaxation
of~\eqref{IP:cover}:
\begin{align}
\begin{array}{ll@{\hspace{10pt}}l}
\max           & \displaystyle
                 \sum_{i \in I} y_i        \\[12pt]
\mbox{s.t.}    & \displaystyle
                 \sum_{i \in C} y_i \leq w_C & \forall C \in \cC \\[12pt]
               & y_i \geq 0                  & \forall i \in I
\end{array}
\tag{D}
\label{LP:D}
\end{align}
Next, for an instance $\cC$ of set cover and a set $S$, we define the
\emph{restriction of $\cC$ to $S$} by $\cC|_S \DEF \set{C\cap S \mid C
\in \cC}$, namely we project out all elements not in $S$. Note that
for any $S$, a solution to $\cC$ is also a solution to $\cC|_S$: we
may only discard some of the constraints in \eqref{IP:cover}.  We now
arrive at our central concept.

\begin{definition}[Dual Obliviousness]
\label{Def:dual}
Let $\Pi$ be a \eqref{IP:cover}-problem.  An algorithm $A$ for $\Pi$
is called 
\emph{$\rho$-dual oblivious} if there exists a constant $\delta$ such
that for every instance $\cC\in\Pi$ there exists a dual solution $y
\in \reals^n$ to \eqref{LP:D} satisfying, for all $S
\subseteq I$, that
\[
A(\cC|_S) \leq \rho \cdot \sum_{i \in S} y_i + \delta~.
\]
\end{definition}

Let us show that the First-Fit (\ff) heuristic for \bp\ is dual
oblivious (we use this property later).  In \ff, the algorithm scans
the items in arbitrary order and places each item in the left most bin
which has enough space to accommodate it, possibly opening a new bin
if necessary.  A newly open bin is placed to the right of rightmost
open bin.

\begin{observation}
\label{obs:ff}
First-Fit is a $2$-dual oblivious  algorithm for bin packing.
\end{observation}
\begin{proof}
In any solution produced by \ff, all non-empty bins except perhaps one
are more than half-full.  Furthermore, this property holds throughout
the execution of \ff, and regardless of the order in which items are
scanned.  It follows that if we let $y_i = a_i$, where $a_i$ is the
size of the $i$th item, then for every $S \subseteq I$ we have $\ff(S)
\leq \max \{2 \sum_{i \in S} y_i, 1\} \leq 2 \sum_{i \in S} y_i +1$,
and hence \ff\ is dual oblivious for \bp\ with $\rho=2$ and
$\delta=1$.
\end{proof}

The usefulness of dual obliviousness is expressed in the following
result.  Let $\Pi$ be a \eqref{IP:cover}-problem, and suppose that
\appr\ is a $\rho$-dual oblivious algorithm for $\Pi$.  Suppose
further that we can efficiently find the dual solution $y$ promised by
dual obliviousness.
Under these assumptions, Algorithm~\ref{Alg:MVBP} below solves any
instance $\cC$ of $\Pi$.

\begin{algorithm}[ht]
\begin{small}
\caption{}
\label{Alg:MVBP}
\begin{algorithmic}[1]
\STATE (\emph{Linear Programming})
       Find an optimal solution $x^*$ to~\eqref{IP:cover}.  Let
       $\opt^*$ denote its value.
\STATE (\emph{greedy phase}) 
        Let $\cC^+ = \set{C ~:~ x^*_C>0}$. Let $\cG\gets\emptyset, S\gets I$.
\WHILE{$\displaystyle\sum_{C\in\cG}w_c<\ln \rho \cdot  \opt^*$}
  \STATE Find $C \in \cC^+$ for which  
         $\displaystyle \inv{w_C} \sum_{i \in S \setminus C} y_i$
         is maximized;
  \STATE $\cG \gets \cG\cup\set{C}$, $S \gets S\setminus C$.
\ENDWHILE
\STATE (\emph{residual solution})
      Apply $\appr$ to the residual instance $S$, obtaining solution
      $\cA$.
\RETURN $\cG \cup \cA$.
\end{algorithmic}
\end{small}
\end{algorithm}

We now bound the weight of the solution $\cG \cup \cA$ that is
computed by Algorithm~\ref{Alg:MVBP}.

\begin{theorem}
\label{The:rna}
Let $\Pi$ be a \eqref{IP:cover}-problem.  Then for any instance of
$\Pi$ with optimal fractional solution $\opt^*$,
Algorithm~\ref{Alg:MVBP} outputs $\cG \cup \cA$ satisfying
\[
w(\cG \cup \cA)
\leq (\ln \rho + 1) \opt^* + \delta + w_{\max}
~,
\]
where $w_{\max} = \max_t w_t$.  
\end{theorem}
\begin{proof}
Clearly, $w(\cG) < \ln \rho \cdot \opt^* + w_{\max}$.  It remains
to bound the weight of $\cA$.
Let $S'$ be the set of items not covered by $\cG$.  We prove that
$\sum_{i \in S'} y_i \leq \frac{1}{\rho} \sum_{i \in I} y_i$, which
implies
\[
w(\cA)
~\leq~ \rho \sum_{i \in S'} y_i + \delta 
~\leq~ \rho e^{-\ln \rho} \sum_{i=1}^n y_i + \delta 
~\leq~ \opt^* + \delta
~,
\]
proving the theorem. 

Let $C_k \in \cC^+$ denote the $k$th subset added to $\cG$ during
the greedy phase, and let $S_k \subseteq I$ be the set of items not
covered after the $k$th subset was chosen.  Define $S_0 = I$.
We prove, by induction on $|\cG|$,  that for every $k$,
\begin{equation}
  \label{eq:ind}
\sum_{i \in S_k} y_i 
\leq \prod_{q=1}^k
       \paren{1-\frac{w_{C_q}}{\opt^*}} \cdot \sum_{i \in I} y_i
\end{equation}
For the base case we have trivially 
$\sum_{i \in S_0} y_i \leq \sum_{i\in I} y_i$.
For the inductive step, assume that 
\[
\sum_{i \in S_{k-1}} \!\! y_i 
\leq \prod_{q=1}^{k-1} 
       \paren{1-\frac{w_{C_q}}{\opt^*}} \cdot \sum_{i \in I} y_i
~.
\]
By the greedy rule and the pigeonhole principle, we have that
\[
\inv{w_{C_k}} \sum_{i \in S_{k-1} \cap C_k} \!\!\!\! y_i 
\geq \inv{\opt^*} \sum_{i \in S_{k-1}} \!\! y_i
~.
\]
It follows that
\[
\sum_{i \in S_k} y_i 
~=~    \sum_{i \in S_{k-1}} \!\! y_i - 
       \sum_{i \in S_{k-1} \cap C_k} \!\!\!\! y_i 
~\leq~ (1 - \frac{w_{C_k}}{\opt^*}) \sum_{i \in S_{k-1}} \!\! y_i
~\leq~ \prod_{q=1}^k
       \paren{1-\frac{w_{C_q}}{\opt^*}} \cdot \sum_{i \in I} y_i
~,
\]
completing the inductive argument. The theorem now follows, since by
\eqref{eq:ind} we have
\[
\sum_{i \in S'} y_i 
\leq \paren{1-\frac{\ln \rho}{k}}^k \cdot \sum_{i \in I} y_i
\leq e^{-\ln \rho} \sum_{i \in I} y_i~,
\]
and we are done.
\end{proof}

Note that if $x^*$ can be found in polynomial time, and if $\appr$ is
a polynomial-time algorithm, then Algorithm~\ref{Alg:MVBP} runs in
polynomial time. Also observe that Theorem~\ref{The:rna} holds even if
$x^*$ is not an optimal solution of~\eqref{IP:cover}, but rather a
$(1+\eps)$-approximation.  We use this fact later.

In this section we defined the notion of \emph{dual obliviousness} of
an algorithm.  We note that Bansal et al.~\cite{BCS06} defined a more
general property of algorithms called \emph{subset obliviousness}.
(For example, a subset oblivious algorithm is associated with several
dual solutions.)
Furthermore, Bansal et al.\ showed that the asymptotic PTAS for \bp\
from~\cite{FerLue81} with minor modifications is subset oblivious and
used it to obtain a subset oblivious $(D+\eps)$-approximation
algorithm for \mvbp.  This paved the way to an algorithm for \vbp,
whose approximation guarantee is arbitrarily close to $\ln D + 1$.
However, in the case of \mvbp, using the above APTAS for \bp\ (at
least in a straightforward manner) would lead to a subset oblivious
algorithm whose approximation guarantee is $(DT + \eps)$.  In the next
section we present a $2D$-dual oblivious algorithm for weighted \mvbp\
that is based on First-Fit.

\subsection{Algorithm for Multiple-Choice Vector Bin Packing}

We now apply the framework of Theorem~\ref{The:rna} to derive an
approximation algorithm for \mvbp.  There are several gaps we need to
fill.

First, we need to solve \eqref{IP:cover} for \mvbp, which consists of
a polynomial number of constraints (one for each item), but an
exponential number of variables.  We circumvent this difficulty as
follows.  Consider the dual of~\eqref{IP:cover}.  The \emph{separation
problem} of the dual program in our case is to find (if it exists) a
subset $C$ with $\sum_{i \in C} y_i > w_C$ for given item profits
$y_1,\ldots,y_n$.  The separation problem can therefore be solved by
testing, for each bin type, whether the optimum is greater than $w_t$,
which in turn is simply an \mmk\ instance, for which we have presented
a PTAS in Section~\ref{Sec:knapsack}.  In other words, the separation
problem of the dual program \eqref{LP:D} has a PTAS, and hence there
exists a PTAS for the LP-relaxation
of~\eqref{IP:cover}~\cite{PST95,GLS88}.

Second, we need to construct a dual oblivious algorithm for \mvbp.  To
do that, we introduce the following notation.  For every item $i \in
I$, incarnation $j$, dimension $d$, and bin type $t$ we define the
\emph{load} of incarnation $j$ of $i$ on the $d$th dimension of bins
of type by $\ell_{ijtd} = a_{ijd}/b_{td}$.
For every item $i \in I$ we define the \emph{effective load} of $i$ as
\[
\bl_i 
= \min_{1 \leq j \leq m, 1 \leq t \leq T}
    \set{w_t \cdot \max_d \ell_{ijtd}}
~.
\]
Also, let $t(i)$ denote the bin type that can contain the most 
(fractional) copies of some incarnation of item $i$, where
$j(i)$ and $d(i)$ are the incarnation and dimension that determine
this bound. Formally:
\begin{align*}
j(i) & = \argmin_j \min_t \{w_t \cdot \max_d \ell_{ijtd} \} \\
t(i) & = \argmin_t \{w_t \cdot \max_d \ell_{ij(i)td} \} \\
d(i) & = \argmax_d \ell_{ij(i)t(i)d} 
~. 
\end{align*}
Assume that $j(i)$, $t(i)$ and $d(i)$ are the choices of $j$, $t$ and
$d$ that are taken in the definition of $\bl_i$.

Our dual oblivious algorithm \appr\ for \mvbp\ is as follows:
\begin{enumerate}
\item Divide the item set  $I$ into $T$ subsets by letting 
      $I_t \DEF \set{i ~:~ t(i)=t}$.
\item Pack each subset $I_t$ in bins of type $t$ using \ff, where the 
      size of each item $i$ is $a_{ij(i)d(i)}$.
\end{enumerate}

Observe that the size of incarnation $j(i)$ of item $i$ in dimension
$d(i)$ is the largest among all other sizes of this incarnation.
Hence, the solution computed by \ff\ is feasible for $I_t$.

We now show that this algorithm is $2D$-dual oblivious.

\begin{lemma}
\label{lemma:subset}
Algorithm \appr\ above is a polynomial time $2D$-dual oblivious
algorithm for \mvbp.
\end{lemma}
\begin{proof}
Consider  an instance of  \mvbp\ with item set $I$, and let the
corresponding set cover problem instance be $\cC$.
We show that there exists a dual solution $y \in \reals^n$ such that
for any $S \subseteq I$,
\[
\appr({\cC}|_S) \leq 2D \cdot \sum_{i \in S} y_i + \sum_{t=1}^T w_t
~.
\]
Define $y_i = \bl_i/D$ for every $i$.  We claim that $y$ is a feasible
solution to~\eqref{LP:D}.
Let $C \in \cC$ be a compatible item set. $C$ induces some bin type
$t$, and an incarnation $j'(i)$ for each $i \in C$. 
Let $d'(i) = \argmax_d \set{a_{ij'(i)d}/b_{td}}$, i.e.,
$d'(i)$ is a dimension of bin type $t$ that receives maximum load from
(incarnation $j'(i)$ of) item $i$.  Then
\begin{align*}
\sum_{i \in C} y_i 
& =    \sum_{d=1}^D \sum_{\substack{i\in C \\ i : d'(i)=d}} 
                      \frac{\bl_i}{D} \\
& \leq \inv{D} 
         \sum_{d=1}^D \sum_{\substack{i\in C \\ i : d'(i)=d}} 
           w_t \cdot \ell_{ij'(i)td} \\
& =    \frac{w_t}{D} 
         \sum_{d=1}^D \sum_{\substack{i\in C \\ i : d'(i)=d}} 
           \frac{a_{ij'(i)d}}{b_{td}} \\
& \leq \frac{w_t}{D} \sum_{d=1}^D \inv{b_{td}} \cdot b_{td} \\
& =    w_t
~,
\end{align*}
where the last inequality follows from the compatibility of $C$.

Now, since $\ff$ computes bin assignments that occupy at most twice
the sum of bin sizes, we have that 
\[
\ff(I_t) 
\leq w_t \cdot \max \set{2\sum_{i \in I_t} \bl_i/w_t, 1} 
\leq 2\sum_{i \in I_t} \bl_i + w_t
~.  
\]
Hence, for every instance $\cal I$ of \mvbp\ we have
\begin{align*}
\appr({\cal I}) 
& =    \sum_{t=1}^T \ff(I_t)\\ 
& \leq \sum_{t=1}^T \paren{2 \sum_{i \in I_t} \bl_i + w_t} \\
& =    2 \sum_{i\in I} \bl_i + \sum_{t=1}^T w_t \\
& =    2D \sum_{i\in I} y_i + \sum_{t=1}^T w_t \\ 
& \leq 2D \cdot \opt^* + \sum_{t=1}^T w_t ~.
\end{align*}
Furthermore, for every $S \subseteq I$ we have
\[
\appr(\cC|_S) 
~=~    \sum_{t=1}^T \ff(S \cap I_t)
~\leq~ 2 \sum_{i \in S} \bl_i + \sum_{t=1}^T w_t
~=~    2D \sum_{i \in S} y_i + \sum_{t=1}^T w_t
~,
\]
and we are done.
\end{proof}

Based on Theorem~\ref{The:rna} and Lemma~\ref{lemma:subset} we obtain
our main result.

\begin{theorem}
\label{The:approx}
If $D = O(1)$, then there exists a polynomial time algorithm for
\mvbp\ with $T$ bin types that computes a solution whose size is at
most
\[
(\ln 2D +1) \opt^* + \sum_{t=1}^T w_t + w_{\max}
~.
\]
\end{theorem}

This implies the following result for unweighted \mvbp:

\begin{corollary}
If $D = O(1)$, then there exists a polynomial time algorithm for
unweighted \mvbp\ with $T$ bin types that computes a solution whose
size is at most
\[
(\ln 2D +1) \opt^* + T + 1
~.
\]
Furthermore, if $T = O(1)$, then there exists a polynomial time $(\ln
2D + 1 + \eps)$-approximation algorithm for unweighted \mvbp, for
every $\eps>0$.
\end{corollary}

We also have the following for weighted \mvbp.

%

\begin{corollary}
If $D = O(1)$ and $T = O(\log n)$, then there exists a polynomial time
$(\ln 2D + 3)$-approximation algorithm for \mvbp.
\end{corollary}
\begin{proof}
The result follows from the fact that as we show, we may assume that
$\sum_t w_t 
\leq \opt$.  In this case, due to Theorem~\ref{The:approx} we have
that the cost of the computed solution is at most
\[
(\ln 2D +1) \opt^* + \sum_{t=1}^T w_t + w_{\max}
\leq (\ln 2D +3) \opt
~.
\]
The above assumption is fulfilled by the following wrapper for our
algorithm:  Guess which bin types are used in some optimal
solution. Iterate through all $2^T-1$ guesses, and for each guess,
compute a solution for the instance that contains only the bin
types in the guess.  Output the best solution.  Since our algorithm
computes a $(\ln2D + 3)$-approximate solution for the right guess, the
best solution is also a $(\ln2D + 3)$-approximation.
\end{proof}

\bigskip

\bibliographystyle{abbrv}
{\small

\begin{thebibliography}{10}

\bibitem{ARKMS06}
M.~M. Akbara, M.~S. Rahmanb, M.~Kaykobadb, E.~Manninga, and G.~Shojaa.
\newblock Solving the multidimensional multiple-choice knapsack problem by
  constructing convex hulls.
\newblock {\em Computers \& Operations Research}, 33:1259--1273, 2006.

\bibitem{BCS06}
N.~Bansal, A.~Caprara, and M.~Sviridenko.
\newblock Improved approximation algorithms for multidimensional bin packing
  problems.
\newblock In {\em 47th {IEEE} Annual Symposium on Foundations of Computer
  Science}, pages 697--708, 2006.

\bibitem{CHW76}
A.~K. Chandra, D.~S. Hirschberg, and C.~K. Wong.
\newblock Approximate algorithms for some generalized knapsack problems.
\newblock {\em Theoretical Computer Science}, 3(3):293--304, 1976.

\bibitem{CheKha04}
C.~Chekuri and S.~Khanna.
\newblock On multidimensional packing problems.
\newblock {\em {SIAM} Journal on Computing}, 33(4):837--851, 2004.

\bibitem{CorEps08}
J.~R. Correa and L.~Epstein.
\newblock Bin packing with controllable item sizes.
\newblock {\em Information and Computation}, 206(8):1003--1016, 2008.

\bibitem{FerLue81}
W.~{Fernandez~de~la~Vega} and G.~S. Lueker.
\newblock Bin packing can be solved within 1+epsilon in linear time.
\newblock {\em Combinatorica}, 1(4):349--355, 1981.

\bibitem{FriLan86}
D.~K. Friesen and M.~A. Langston.
\newblock Variable sized bin packing.
\newblock {\em {SIAM} Journal on Computing}, 15(1):222--230, 1986.

\bibitem{FriCla84}
A.~M. Frieze and M.~R.~B. Clarke.
\newblock Approximation algorithms for the $m$-dimensional $0-1$ knapsack
  problem: worst-case and probabilistic analyses.
\newblock {\em European Journal of Operational Research}, 15:100--109, 1984.

\bibitem{GGJY76}
M.~R. Garey, R.~L. Graham, D.~S. Johnson, and A.~C. Yao.
\newblock Resource constrained scheduling as generalized bin packing.
\newblock {\em J. Comb. Theory, Ser. A}, 21(3):257--298, 1976.

\bibitem{GLS88}
M.~Gr\"otschel, L.~Lovasz, and A.~Schrijver.
\newblock {\em Geometric Algorithms and Combinatorial Optimization}.
\newblock Springer-Verlag, 1988.

\bibitem{HMS-04}
M.~Hifi, M.~Michrafy, and A.~Sbihi.
\newblock Heuristic algorithms for the multiple-choice multidimensional
  knapsack problem.
\newblock {\em Journal of the Operational Research Society}, 55:1323--1332,
  2004.

\bibitem{KarKar82}
N.~Karmarkar and R.~M. Karp.
\newblock An efficient approximation scheme for the one-dimensional bin-packing
  problem.
\newblock In {\em 23rd {IEEE} Annual Symposium on Foundations of Computer
  Science}, pages 312--320, 1982.

\bibitem{knapsack-book}
H.~Kellerer, U.~Pferschy, and D.~Pisinger.
\newblock {\em Knapsack Problems}.
\newblock Springer, Berlin, 2004.

\bibitem{Khan}
M.~S. Khan.
\newblock {\em Quality Adaptation in a Multisession Multimedia System: Model,
  Algorithms and Architecture}.
\newblock PhD thesis, Dept. of Electrical and Computer Engineering, 1998.

\bibitem{MC-84}
M.~J. Magazine and M.-S. Chern.
\newblock A note on approximation schemes for multidimensional knapsack
  problems.
\newblock {\em Mathematics of Operations Research}, 9(2):244--247, May 1984.

\bibitem{PapadimitriouS}
C.~H. Papadimitriou and K.~Steiglitz.
\newblock {\em Combinatorial optimization : algorithms and complexity}.
\newblock Prentice-Hall, 1981.

\bibitem{PD-05}
R.~{Parra-Hern\'andez} and N.~J. Dimopoulos.
\newblock A new heuristic for solving the multichoice multidimensional knapsack
  problem.
\newblock {\em IEEE Trans. on Systems, Man, and Cybernetics---Part A: Systems
  and Humans}, 35(5):708--717, 2005.

\bibitem{PST95}
S.~A. Plotkin, D.~B. Shmoys, and {\'E}.~Tardos.
\newblock Fast approximation algorithms for fractional packing and covering
  problems.
\newblock {\em Mathematics of Operations Research}, 20:257--301, 1995.

\bibitem{S-07}
A.~Sbihi.
\newblock A best first search exact algorithm for the multiple-choice
  multidimensional knapsack problem.
\newblock {\em Journal of Combinatorial Optimization}, 13(4):337--351, May
  2007.

\bibitem{SvSE03}
S.~S. Seiden, R.~van Stee, and L.~Epstein.
\newblock New bounds for variable-sized online bin packing.
\newblock {\em {SIAM} Journal on Computing}, 32(2):455--469, 2003.

\bibitem{ShaTam03}
H.~Shachnai and T.~Tamir.
\newblock Approximation schemes for generalized 2-dimensional vector packing
  with application to data placement.
\newblock In {\em 6th International Workshop on Approximation Algorithms for
  Combinatorial Optimization Problems}, number 2764 in LNCS, pages 129--148,
  2003.

\bibitem{WP-72}
D.~Warner and J.~Prawda.
\newblock A mathematical programming model for scheduling nursing personnel in
  a hospital.
\newblock {\em Manage. Sci. (Application Series Part 1)}, 19:411--422, Dec.
  1972.

\end{thebibliography}

}

\end{document}